\documentclass[12pt]{article}
\usepackage[left=3cm, right=3cm, top=3cm, bottom=3cm]{geometry}
\usepackage{amsmath,amsthm,amssymb,natbib,enumerate,xcolor}
\usepackage[hyphens]{url}
\usepackage{hyperref}
\hypersetup{
    colorlinks,
    linkcolor={red!50!black},
    citecolor={black},
    urlcolor={black}
}
\newtheorem{thm}{Theorem}
 
\newtheorem{lemma}[thm]{Lemma}

\theoremstyle{definition}

\theoremstyle{remark}

\newcommand*{\set}[1]{\left\{#1\right\}}
\newcommand*{\good}{G}
\newcommand*{\bad}{B}
\newcommand*{\x}{X}
\newcommand*{\y}{Y}

\title{Price competition with uncertain quality and cost}
\author{Sander Heinsalu\thanks{Research School of Economics, Australian National University, 
25a Kingsley St, Acton ACT 2601, Australia.
Email: sander.heinsalu@anu.edu.au, 
website: \url{https://sanderheinsalu.com/}
}}
\date{\today}

\begin{document}
\maketitle

\begin{abstract}
Consumers in many markets are uncertain about firms' qualities and costs, so buy based on both the price and the quality inferred from it. Optimal pricing depends on consumer heterogeneity only when firms with higher quality have higher costs, regardless of whether costs and qualities are private or public. If better quality firms have lower costs, then good quality is sold cheaper than bad under private costs and qualities, but not under public. However, if higher quality is costlier, then price weakly increases in quality under both informational environments. 

	
\end{abstract}


In many markets, buyers are uncertain about the qualities and costs of sellers. In that case, purchasing decisions depend on both the price and the quality that the buyers infer from price (as opposed to the actual quality). This paper shows that the pricing decisions of the firms differ based on customer heterogeneity when higher quality producers have greater marginal cost, 
irrespective of whether costs and qualities are private or public. 

If firms privately know their quality and cost, and a higher quality firm has a smaller marginal cost,\footnote{
Many industries have higher quality associated with a lower cost, e.g.\ mutual funds \citep{gil-bazo+ruiz-verdu2009}, cotton weaving \citep{bloom+2013}, medical innovations \citep{nelson+2009}. Additional empirical examples and theoretical reasons for negatively correlated quality and cost are in \cite{heinsalu2019a}. 
} 
then it signals its quality by pricing lower than a bad quality rival. By contrast, with public qualities and costs, better quality is priced higher. 
In markets where higher quality providers have larger costs, the price is constant or increases in quality under both informational environments. 

In the markets considered in this paper, competing firms independently draw a type, either \emph{good} or \emph{bad}. The good type has better quality than the bad, and in the main model also higher marginal cost. Consumers have heterogeneous valuations for the firms' products, with a greater valuation also implying a weakly larger premium for quality. Each player knows her own type, but other players only have a common prior over the types. The firms simultaneously set prices, which the consumers observe. Then each consumer chooses either to buy from one of the firms or leave the market. 
Consumers Bayes update their beliefs about the types of the firms based on the prices. 
The equilibrium notion is perfect Bayesian equilibrium (PBE). 

In equilibrium, price is above the competitive level, regardless of whether the good and bad types pool or the good quality firms signal their type by raising price. 
Prices in pooling and semipooling equilibria exceed the cost of the good type, thus both types of at least one firm make positive profit. In separating equilibria, both firms' bad types obtain positive profit. 

By contrast, complete information Bertrand competition between identical firms, whether good or bad quality, leads to zero profit and a lower price than under incomplete information. Complete information competition between a good and a bad firm may raise price compared to private cost and quality, but one of the types still makes zero profit. 
The \emph{ex ante} expected price may be higher or lower under complete information. 
The \emph{ex ante} price dispersion under public types exceeds that under private. 
If the cost and quality differences between the types vanish or the probability of the good type goes to zero, then the complete and incomplete information environments converge to the same price: the marginal cost of the bad type. 
The outcome in this paper is independent of whether the firms observe each other's cost or quality, but relies on consumers not observing these. 

The equilibrium pricing differs from a privately informed monopolist, and from competition when consumers find it costly to learn prices. 
A monopolist with good quality signals its type by a high price. 
When observing the prices of competing firms is costly, the outcome under complete information is monopoly pricing \citep{diamond1971}. Incomplete information leads to above-monopoly pricing in costly search, except when quality and cost are negatively correlated, in which case pricing is competitive \citep{heinsalu2019a}. 

\textbf{Literature.} 
The signalling literature started from \cite{spence1973}, and price as a signal was studied in \cite{milgrom+roberts1986}. In the present paper, the consumers are the receivers of the price signal, unlike in limit pricing (\cite{milgrom+roberts1982b} and the literature following) where the incumbent deters potential entrants from entering the market by signalling its low cost via a low price. 

Bertrand competition has been combined with price signalling in \cite{janssen+roy2015}, where consumers are homogeneous and firms may verifiably disclose their types. 
In \cite{sengupta2015}, consumers may or may not value quality, but are otherwise homogeneous. Bertrand competing firms publicly invest to obtain a random private quality improvement, and signal quality via price. 

\cite{hertzendorf+overgaard2001} assume that one firm has high and the other low quality (firm types are perfectly negatively correlated, thus firms know each other's type) and that cost does not depend on quality. Firms may signal via price or advertising. Full separation requires advertising. 

\cite{hertzendorf+overgaard2001b,daughety+reinganum2007,daughety+reinganum2008b} consider Hotelling competition with quality differences (thus both horizontal and vertical differentiation). \cite{daughety+reinganum2007,daughety+reinganum2008b} focus on symmetric separating equilibria. \cite{hertzendorf+overgaard2001b} show the nonexistence of full separation, similarly to the current work. 


If firm types only differ in their private marginal cost, but not quality, then the high cost type prices at its marginal cost, but the low cost type mixes over a range of prices strictly above its marginal cost and weakly below the price of the high-cost type \citep{spulber1995}. 

The next section sets up the model and Section~\ref{sec:results} characterises the equilibrium set, first when cost and quality are positively associated. Negative correlation is examined in Section~\ref{sec:negcorr}. 

\section{Model}
\label{sec:model}

Two firms indexed by $i\in\set{\x,\y}$ compete. Each draws an i.i.d.\ type $\theta\in\set{\good,\bad}$, interpreted respectively as good and bad, with $\Pr(\good)=\mu_0\in(0,1)$. A continuum of consumers of mass $1$, indexed by $v\in[0,\overline{v}]$, is distributed according to the strictly positive continuous pdf $f_v$, with cdf $F_v$. Consumer types are independent of firm types. 
All players know their own type, but only have a common prior belief over the types of other players. 

A type $\theta$ firm has marginal cost $c_{\theta}$, with $c_{\good}>c_{\bad}>0$. A type $\good$ firm has better quality: a type $v$ consumer values a type $\bad$ firm's product at $v$ and $\good$'s product at $h(v)\geq v$, with $h'\geq 1$ and $h(\overline{v})>\overline{v}$. For some results, $h$ is restricted to the form $h(v)=v+\nu$, with $\nu>0$. In this case, all consumers are willing to pay the same premium for quality. 
Assume $\overline{v}>c_{\good}\geq h(0)$, so demand for the good type firm is positive. The previous assumption $c_{\bad}>0$ ensures that not all consumers buy from $\bad$ under complete information. 
Firms and consumers are risk-neutral. Consumers have unit demand. 

The firms observe their types and simultaneously set prices $P_{\x},P_{\y}\in [0,P_{\max}]$, where $P_{\max}\in(h(\overline{v}),\infty)$. 
A behavioural strategy of firm $i$ maps its type to $\Delta \mathbb{R}_{+}$.\footnote{Denote by $\Delta S$ the set of probability distributions on the set $S$. 
} 
The probability that firm $i$'s type $\theta$ assigns to prices below $P$ is denoted $\sigma_i^{\theta}(P)$, so $\sigma_i^{\theta}(\cdot)$ is the cdf of price. 
The corresponding pdf is denoted $\frac{d\sigma_i^{\theta}(P)}{dP}$ if it exists.

After seeing the prices the firms, a consumer decides whether to buy from firm $\x$ (denoted $b_{\x}$), firm $\y$ ($b_{\y}$) or not at all ($n$). 
The behavioural strategy $\sigma:[0,\overline{v}]\times \mathbb{R}_{+}^2\rightarrow \Delta\set{b_{\x},b_{\y},n}$ of a consumer maps his valuation and the prices to a decision. 

The \emph{ex post} payoff of a type $\theta$ firm if it sets price $P$ and a mass $D$ of consumers buy from it is $(P-c_{\theta})D$. 
Define $v_{\mu}(P):=\inf\{v:\mu h(v)+(1-\mu)v\geq P\}$. In particular, $v_{1}(\cdot)=h^{-1}(\cdot)$ and $v_{0}(\cdot)=id(\cdot)$. Total demand at price $P$ and a fixed posterior belief $\mu$ of the consumers is $D^{\mu}(P)=1-F_{v}(v_{\mu}(P))$. 
If $h(v)=v+\nu$, then $D^{\mu}(P)=1-F_{v}(P-\mu \nu)$. 
The monopoly profit of type $\theta$ at $P,\mu$ is $\pi_{\theta,\mu}^{m}(P) =(P-c_{\theta})D^{\mu}(P)$. 
The complete-information monopoly profit $\pi_{\bad,0}^{m}(P)$ is denoted $\pi_{\bad}^{m}(P)$, and $\pi_{\good,1}^{m}(P)$ denoted $\pi_{\good}^{m}(P)$. The monopoly price is $P_{\theta}^{m}:=\arg\max_{P}\pi_{\theta}^{m}(P)$. 
Assume that $\pi_{\theta}^{m}(P)$ is single-peaked in $P$. 
To avoid trivial separation, assume that $P_{\bad}^{m}\geq c_{\good}$ or $\pi_{\bad}^{m} <(c_{\good}-c_{\bad})D^{1}(c_{\good})$. 

A consumer's belief about firm $i$ conditional on price $P$ and the firm's trategy $\sigma_i^*$ is  
\begin{align}
\label{mu}
\mu_i(P) :=\frac{\mu_0\frac{d}{dP}\sigma_i^{\good*}(P)}{\mu_0\frac{d}{dP}\sigma_i^{\good*}(P) +(1-\mu_0)\frac{d}{dP}\sigma_i^{\bad*}(P)} 
\end{align}
whenever $\mu_0\frac{d}{dP}\sigma_i^{\good*}(P) +(1-\mu_0)\frac{d}{dP}\sigma_i^{\bad*}(P)>0$. 
A discontinuity of height $h_{\theta}$ in the cdf $\sigma_i^{\theta*}$ is interpreted in the pdf as a Dirac $\delta$ function times $h_{\theta}$, thus makes the denominator of~(\ref{mu}) positive. 
A jump in $\sigma_i^{\good*}(\cdot)$, but not $\sigma_i^{\bad*}(\cdot)$ at $P$ yields $\mu_i(P)=1$, and a jump in $\sigma_i^{\bad*}(\cdot)$, but not $\sigma_i^{\good*}(\cdot)$ results in $\mu_i(P)=0$. If each $\sigma_{i}^{\theta *}$ has an atom of respective size $h_{\theta}$ at $P$, then $\mu_i(P)=\frac{\mu_0h_{\good}}{\mu_0h_{\good}+(1-\mu_{0})h_{\bad}}$. 
Finally, if the denominator of~(\ref{mu}) vanishes, then choose an arbitrary belief. 

The equilibrium notion is perfect Bayesian equilibrium (PBE), henceforth simply called equilibrium: each player maximises its payoff given its belief about the strategies of the others, and the beliefs are derived from Bayes' rule when possible. 

The minimal and maximal price of firm $i$'s type $\theta$ in a candidate equilibrium are denoted $\underline{P}_{i\theta}$ and $\overline{P}_{i\theta}$, respectively. 


\section{Results}
\label{sec:results}

First the benchmark of complete information is considered, which illustrates some general features of the framework. After that, the case of private cost and quality is examined. 

One general observation is that among the consumers who end up buying, the ones with a low valuation for good quality relative to bad sort to firms believed to have lower quality, while the high valuation consumers go to expected high quality. If the quality difference between the firm types is large compared to their cost difference and firms draw unequal types, then the low quality firm is left with zero demand. Similarly, if firm types differ and the variation in quality is small and in cost large, then the high quality firm receives no customers.

\subsection{Benchmark: complete information} 
\label{sec:completeinfo}
Symmetric firms with publicly known qualities price at their marginal cost, regardless of whether consumers are homogeneous or not and whether their quality premium $h(v)-v$ is constant or increasing. 

With asymmetric firms and a constant quality premium, one type prices at its marginal cost and the other higher by just enough to make the consumers indifferent. For example, if $c_{\good}-c_{\bad}>\nu$, then $P_{\good}=c_{\good}$ and $P_{\bad}=c_{\good}-\nu$, but if $c_{\good}-c_{\bad}<\nu$, then $P_{\good}=c_{\bad}+\nu$ and $P_{\bad}=c_{\bad}$. Given the chance, all consumers leave the type with relatively higher cost for its quality (if $c_{\good}-c_{\bad}>\nu$, then $\good$, otherwise $\bad$), because otherwise the type with the lower relative cost would undercut slightly. 
The outcome of one type pricing at its marginal cost and receiving zero demand is standard in asymmetric Bertrand competition. 

In asymmetric price competition when the quality premium increases in consumer valuation, define the indifferent consumer as $v^{*}(P_{\good}-P_{\bad}):=\inf\{v\geq 0: h(v)-v\geq P_{\good}-P_{\bad}\}$. Assume $h'>1$, which ensures $v^{*}$ is a strictly increasing continuous function. 
Consumers with valuations above $v^{*}(P_{\good}-P_{\bad})$ buy from $\good$, those below from $\bad$ or not at all---the latter if $v<v_0(P_{\good},P_{\bad}):=\min\set{P_{\bad},h^{-1}(P_{\good})}$. 
A good type firm who expects $P_{\bad}^*$ solves $\max_{P_{\good}}(P_{\good}-c_{\good})[1-F_{v}(v^*(P_{\good}-P_{\bad}^*))]$ and a bad type solves $\max_{P_{\bad}}(P_{\bad}-c_{\bad})[F_{v}(v^*(P_{\good}^*-P_{\bad})) -F_{v}(P_{\bad})]$. 
Equilibrium exists, because prices may w.l.o.g.\ be restricted to the convex compact interval $[0,h(\overline{v})]$ and payoffs are continuous in the action profiles. 

If $h(c_{\bad})-c_{\good}< 0=c_{\bad}-c_{\bad}$, then type $\bad$ obtains positive demand even if $\good$ prices at its marginal cost $c_{\good}$. Symmetrically, if $h(\overline{v})-c_{\good}> \overline{v}-c_{\bad}$, then $\good$ can attract some customers facing $P_{\bad}=c_{\bad}$. Under these conditions, each type prices above its marginal cost and obtains positive demand, unlike with homogeneous consumers or a constant quality premium.

\subsection{Incomplete information Bertrand competition} 
\label{sec:incompleteinfo}
Results for the general case of an increasing quality premium for the good type are presented first, followed in Section~\ref{sec:const} by derivations that require the additional restriction of a constant quality premium. 
The first lemma proves that in any equilibrium, if demand for a firm's bad type is positive, then  the bad type prices below the good and obtains greater demand, strictly so under (partial) separation. The proof combines the incentive constraints (ICs) of the types and is standard. 
\begin{lemma}
\label{lem:monot}
In any equilibrium, $D^{\mu_{i}(P_{i\bad})}(P_{i\bad})\geq D^{\mu_{i}(P_{i\good})}(P_{i\good})$, and if $D^{\mu_{i}(P_{i\bad})}(P_{i\bad}) >0$, then $P_{i\good}\geq P_{i\bad}$, with equality iff $ D^{\mu_{i}(P_{i\bad})}(P_{i\bad})= D^{\mu_{i}(P_{i\good})}(P_{i\good})$. If $D^{\mu_{i}(P_{i\theta})}(P_{i\theta}) >0$, then $P_{i\theta}\geq c_{\theta}$. 
\end{lemma}
\begin{proof}
Denote $D^{\mu_{i}(P_{i\theta})}(P_{i\theta})$ by $D_{\theta}$ to simplify notation. 
In any equilibrium, for any $P_{i\theta}$ in the support of $\sigma_{i}^{\theta*}$, the incentive constraints IC$_{\good}$: $(P_{i\good}-c_{\good})D_{\good} \geq (P_{i\bad}-c_{\good})D_{\bad}$ and IC$_{\bad}$: $(P_{i\bad}-c_{\bad})D_{\bad}\geq (P_{i\good}-c_{\bad})D_{\good}$ hold. 
Rewrite the ICs as $P_{i\good}D_{\good}-P_{i\bad}D_{\bad} \geq c_{\good}D_{\good}-c_{\good}D_{\bad}$ and $c_{\bad}D_{\good}-c_{\bad}D_{\bad}\geq P_{i\good}D_{\good}-P_{i\bad}D_{\bad}$. Using $c_{\theta}>0$, the ICs become $\frac{P_{i\good}D_{\good}-P_{i\bad}D_{\bad}}{c_{\good}} \geq D_{\good}-D_{\bad}\geq \frac{P_{i\good}D_{\good}-P_{i\bad}D_{\bad}}{c_{\bad}}$. Then from $c_{\good}>c_{\bad}$, we get $D_{\good}\leq D_{\bad}$ and $P_{i\good}D_{\good}\leq P_{i\bad}D_{\bad}$. 

In any equilibrium, if $D_{\theta} >0$, then $P_{i\theta}\geq c_{\theta}$. Thus if $D_{\good} >0$, then $D_{\bad} >0$, otherwise $\bad$ deviates to $P_{i\good}\geq c_{\good}>c_{\bad}$ to get positive profit. 
If $D_{\bad}= D_{\good}>0$, then the ICs imply $P_{i\good}= P_{i\bad}$. The converse implication is obvious. 
%
%
\end{proof}
An implication of Lemma~\ref{lem:monot} is that the supports of $\sigma_{i}^{\bad*}$ and $\sigma_{i}^{\good*}$ have at most one price in common. Thus if one type at a firm semipools, i.e.\ only sets prices that the other type also chooses, then the semipooling type plays a pure strategy. 

The next lemma rules out some equilibria even when belief threats are possible. 
\begin{lemma}
\label{lem:withbeliefthreats}
In any equilibrium, at least one firm's type $\bad$ obtains positive demand. If both types of firm $i$ get zero demand, then both types of $j$ set $P_{j}>c_{\good}$ and receive positive profit. 
If $\mu_0 h(c_{\bad})+(1-\mu_0)c_{\bad}-c_{\good}\leq 0$ or type $\bad$ of a profitable firm partly separates or firms play symmetric strategies, then both firms' $\bad$ types obtain positive profit. 
\end{lemma}
\begin{proof}
Suppose both types of both firms get zero demand in equilibrium. Then $P>\overline{v}$ for all prices. Both types deviate to $P\in(c_{\good},\overline{v})$ to obtain positive demand and profit even at the worst belief $\mu_{i}(P)=0$. 
If firm $i$'s type $\good$ obtains positive demand, then any $P_{i\good}$ that $\good$ sets is above $c_{\good}$, otherwise $\good$ would deviate to $P=c_{\good}$. Demand and profit are positive for $i\bad$, otherwise $i\bad$ would deviate to $P_{i\good}$. 

Suppose both types of firm $j$ receive zero demand. Total demand at $P<\overline{v}$ is positive at any belief, so both types of $i$ obtain positive demand and profit at any $P_{i}\in(c_{\good},\overline{v})$. Positive $\pi_{i\theta}^{*}$ implies $P_{i\theta}>c_{\theta}$. If $i\bad$ partly separates, i.e.\ sets $P_{i\bad}$ with $\mu_{i}(P_{i\bad})=0$, then $j\bad$ can get positive demand and profit by setting $P_{j}\in(c_{\bad},P_{i\bad})$, regardless of $\mu_{j}(P_{j})$. 

If $i\bad$ (semi)pools with $i\good$, i.e.\ only sets prices that type $\good$ also sets, 
then there exists $P_{i0}$ in the support of $\sigma_{i\good}$ s.t.\ $\mu_{i}(P_{i0})\leq \mu_{0}$ and $P_{i0}>c_{\good}$. 
If $\mu_0 h(c_{\bad})+(1-\mu_0)c_{\bad}-c_{\good}\leq 0$, then for $\epsilon>0$ small, consumers with valuations $v\in(c_{\bad},c_{\bad}+\epsilon)$ strictly prefer to buy from $j$ at $P_{j}\in(c_{\bad},v)$ and any belief, rather than from $i$ at $P_{i0}>c_{\good}$ and $\mu_{i}(P_{i0})\leq \mu_{0}$. This makes $j\bad$ deviate to $P_{j}\in(c_{\bad},v)$ to get positive profit. 

Symmetric strategies imply that the firms split the total demand on average. Each firm's type $\good$ receives positive demand with probability at least $\mu_0$ (when the other firm has type $\good$). Then $P_{\good}\geq c_{\good}$. Type $\bad$ can imitate $P_{\good}$, thus $\pi_{i\bad}^{*}>0$. 
\end{proof}
The results of Lemma~\ref{lem:withbeliefthreats} are tight, in the sense that there exist equilibria where one firm obtains zero demand, or both firms' good types zero profit. 
An equilibrium in which the good types receive zero profit is symmetric pooling on $P_{0}=c_{\good}$, which exists if $\mu_0 h(c_{\bad})+(1-\mu_0)c_{\bad}> c_{\good}$ and only if a weak inequality holds. 

Equilibria where one firm obtains zero demand require $\mu_0 h(c_{\bad})+(1-\mu_0)c_{\bad}>c_{\good}$ and asymmetric strategies. For example, firm $\x$ pools on some $P_{\x0}\in[c_{\good},\mu_0 h(c_{\bad})+(1-\mu_0)c_{\bad})$ and firm $\y$ on some $P_{\y0}\in(P_{\x0},P_{\x0}+\epsilon)$ for $\epsilon>0$ small. 
Belief at $P_{i0}$, $i\in\set{\x,\y}$ is $\mu_{i}(P_{i0})=\mu_0$, at other prices zero. Firm $\y$ gets zero demand on or off the equilibrium path from prices $P\geq c_{\bad}$, because if $\mu_0 h(c_{\bad})+(1-\mu_0)c_{\bad}-P_{\x0}>c_{\bad}-c_{\bad}=0$, then all consumers with $v\geq c_{\bad}$ prefer to buy at $P_{\x0}$ and $\mu_{\x}(P_{\x0})=\mu_0$, rather than at $P=c_{\bad}$ and $\mu_{\x}(P)=0$. 
Neither firm wants to deviate, because $P\neq P_{i0}$ results in the worst belief and zero demand. 
Weakly dominated strategies are not played in this equilibrium, provided $P_{\x0}>c_{\good}$. 

Some non-pooling equilibria also involve zero demand for one firm---modify the preceding example so that each type of firm $\y$ (partly) separates, for example $\bad$ sets $P_{\y\bad}>\overline{v}$ and $\good$ sets $P_{\y\good}>h(\overline{v})$. 
To ensure zero demand for $\y$, any $P_{\y\theta}$ receiving positive probability must satisfy 
$P_{\y\theta}-P_{i0}\geq [\mu_{\y}(P_{\y\theta})-\mu_0][h(v)-v]$ for all $v$, which holds if $P_{\y\theta}-[\mu_{\y}(P_{\y\theta})-\mu_0][h(\overline{v})-\overline{v}]\geq P_{i0}$ for any $\mu_{\y}(P_{\y\theta})\geq \mu_0$.\footnote{
The symmetric requirement that if $\mu_{\y}(P_{\y\theta})\leq \mu_0$, then $P_{\y\theta}+[\mu_0-\mu_{\y}(P_{\y\theta})][h(c_{\bad})-c_{\bad}]\geq P_{i0}$ holds whenever $\mu_0 h(c_{\bad})+(1-\mu_0)c_{\bad}>c_{\good}$, which by Lemma~\ref{lem:withbeliefthreats} is necessary for one firm to obtain zero demand. 
} 
An additional necessary condition for equilibria with $P_{\y\bad}>\overline{v}$ and $P_{\y\good}>h(\overline{v})$ is that firm $\x$ as a monopolist does not deviate: $(P-c_{\theta})D^{0}(P)\leq (P_{\x0}-c_{\theta})D^{\mu_{0}}(P_{\x0})$ for any $\theta$ and $P<P_{\y\bad}$. 
This condition is easier to satisfy for smaller $P_{\y\theta}$. Reducing $P_{\y\theta}$ makes the requirement $P_{\y\theta}-P_{i0}\geq [\mu_{\y}(P_{\y\theta})-\mu_0][h(\overline{v})-\overline{v}]$ harder to satisfy. 

Asymmetric pooling passes the Intuitive Criterion under a constant quality premium for some parameter values, as shown in the following lemma. 
\begin{lemma}
\label{lem:asypool}
If $h(v)-v=\nu$ for all $v$, then the Intuitive Criterion does not eliminate equilibria in which firm $i$ pools on some $P_{i0}\in(c_{\good},c_{\bad}+\mu_0 \nu)$ and firm $j$ on some $P_{j0}>P_{i0}$ such that $[P_{j0}+(1-\mu_0)\nu-c_{\bad}][1-F_{v}(P_{j0}-\mu_0\nu)] >(P_{i0}-c_{\bad})[1-F_{v}(P_{i0}-\mu_0\nu)]$. 
\end{lemma} 
\begin{proof}
Let $P_{\y0}>P_{\x0}$ w.l.o.g., so that $D_{\y}(P_{\y0})=0$. 
Asymmetric pooling equilibria exist only if $\mu_0 h(c_{\bad})+(1-\mu_0)c_{\bad}\geq c_{\good}$, by Lemma~\ref{lem:withbeliefthreats}. 
If setting belief to $\mu_{\y}(P_{\y})=1$ for some $P_{\y}$ results in $D_{\y}(P_{\y})>0$, then both types of firm $\y$ deviate, but if $D_{\y}(P_{\y})=0$, then neither type does. Thus for the firm with zero demand, the Intuitive Criterion has no effect. 

Consider the firm $\x$ with $D_{\x}(P_{\x0})>0$ and assume $h(v)-v=\nu$ for all $v$ (constant quality premium), in which case $\mu_0 h(c_{\bad})+(1-\mu_0)c_{\bad}\geq c_{\good}$ 
becomes $\mu_0 \nu\geq c_{\good}-c_{\bad}$. Define $P_{\x1}:=P_{\y0}+(1-\mu_0)\nu$, 
which is the price that makes all consumers indifferent between firms $\x$ and $\y$ at beliefs $\mu_{\x}=1$ and $\mu_{\y}=\mu_0$. Setting belief to $\mu_{\x}(P_{\x})=1$ for any $P_{\x}>P_{\x1}$ does not attract either type of firm $\x$ to deviate, because even at the best belief, all consumers switch to firm $\y$. 

If type $\good$ prefers $P<P_{\x0}$ to $P_{\x0}$, then so does type $\bad$. 

Setting belief to $\mu_{\x}(P_{\x})=1$ for any $P_{\x}\in(P_{\x0}, P_{\x1}]$ makes the bad type of $\x$ deviate to $P_{\x}$ if $(P_{\x1}-c_{\bad})[1-F_{v}(P_{\x1}-\nu)] >(P_{\x0}-c_{\bad})[1-F_{v}(P_{\x0}-\mu_0\nu)]$. This is equivalent to $(P_{\y0}+(1-\mu_0)\nu-c_{\bad})[1-F_{v}(P_{\y0}-\mu_0\nu)] >(P_{\x0}-c_{\bad})[1-F_{v}(P_{\x0}-\mu_0\nu)]$, which holds for $P_{\y0}-P_{\x0}$ small enough, because $F_{v}$ is continuous by assumption. 
Thus there is no price $P$ such that under the best belief, type $\good$ wants to deviate to $P$, but type $\bad$ prefers not to. 
\end{proof}
Symmetric pooling on $P_{0}\in[c_{\good},c_{\bad}+(1-\mu_0)\nu)$ exists if $c_{\good}<c_{\bad}+(1-\mu_0)\nu$ and passes the Intuitive Criterion by reasoning similar to Lemma~\ref{lem:asypool}. Take $P_{\x0}=P_{\y0}=P_{0}$. If the firms set the same price, then assume they split the market equally. The bad type deviates whenever the good type does if $\frac{1}{2}(P_{0}+(1-\mu_0)\nu-c_{\bad})[1-F_{v}(P_{0}-\mu_0\nu)] >(P_{0}-c_{\bad})[1-F_{v}(P_{0}-\mu_0\nu)]$, i.e.\ $(1-\mu_0)\nu>P_{0}-c_{\bad}$. 


In contrast to the constant quality premium case, the Intuitive Criterion eliminates all pooling equilibria when the quality premium strictly increases in the consumer's valuation. 
The idea of the proof relies on profit at a fixed belief being continuous when the quality premium is strictly increasing. At the best belief (certainty of the good type), the bad type prefers to deviate to a price just above pooling and prefers not to deviate to a high enough price, so by the Mean Value Theorem, there is a price at which the bad type is indifferent. At prices just above this indifference and the best belief, the good type still strictly prefers to deviate, which justifies the belief and eliminates the candidate equilibrium. 

\begin{lemma}
\label{lem:elimpool3}
If $h'(v)>1$, then the Intuitive Criterion rules out any equilibria where the support of $\sigma_{i\theta}^{*}$ includes $P_{i0}$ with $\mu_{i}(P_{i0})\in(0,1)$ and $D_{i}(P_{i0})>0$ and the support of $\sigma_{j\good}^{*}$ does not include $P_{j}$ s.t.\ $\mu_{j}(P_{j})=1$. 
\end{lemma}
\begin{proof}
Recall that $v_{\mu}(P):=\inf\{v:\mu h(v)+(1-\mu)v\geq P\}$. For any $P_{j}$, firm $i$'s demand 
\begin{align*}
&D^{\mu}_{i}(P) 
=\int_{v_{\mu}(P)}^{\overline{v}}\text{\textbf{1}}\set{\mu h(v)+(1-\mu)v-P >\mu_{j}(P_j)h(v)+(1-\mu_j(P_{j}))v-P_{j} } 
\\&\notag+\frac{1}{2}\text{\textbf{1}}\set{\mu h(v)+(1-\mu)v-P =\mu_{j}(P_j)h(v)+(1-\mu_j(P_{j}))v-P_{j} }dF_{v}(v)
\end{align*} 
decreases in $P$ and increases in $\mu$, strictly if $D^{\mu}_{i}(P)\in(0,1)$. 
Thus for any distribution $\mu_0\sigma_{j\good}^{*}+(1-\mu_0)\sigma_{j\bad}^{*}$ over $P_{j}$, firm $i$'s expected demand decreases in $P$ and increases in $\mu$, strictly if $\mathbb{E}D^{\mu}_{i}(P) \in(0,1)$. 

If $h'>1$, then $v_{\mu}(P)$ is continuous in $\mu,P$, and strictly increases in $P$ and strictly decreases in $\mu$ for $P\in[\mu h(0),\mu h(\overline{v})+(1-\mu)\overline{v}]$. If $h'>1$, then the inverse of $h(v)-v$ is continuous and on $[h(0),h(\overline{v})-\overline{v}]$, strictly increases in its argument. Further, if $h'>1$ and $\mu\neq \mu_j(P_{j})$, then the indifference $\mu h(v)+(1-\mu)v-P =\mu_{j}(P_j)h(v)+(1-\mu_j(P_{j}))v-P_{j}$ holds for at most one $v\in[0,\overline{v}]$. 
Therefore if $h'>1$ and $\mu\neq \mu_j(P_{j})$, then 
\begin{align*}
&D^{\mu}_{i}(P) 
=\int_{v_{\mu}(P)}^{\overline{v}}\text{\textbf{1}}\set{[\mu-\mu_{j}(P_{j})] [h(v)-v]>P-P_{j} }dF_{v}(v)
\end{align*} 
is continuous in $P,\mu$. In that case, for any distribution over $P_{j}$, the expected demand and profit are continuous in $P,\mu$. 

Suppose w.l.o.g.\ that firm $\x$ does not fully separate and gets positive demand, i.e.\ there exists $P_{\x0}$ chosen by both types of $\x$, with $D_{\x}(P_{\x0})>0$. Being chosen by both types implies $\mu_{\x}(P_{\x0})\in(0,1)$ and by Lemma~\ref{lem:monot}, $P_{\x0}\geq c_{\good}$ and is unique. The assumption $c_{\good}>h(0)$ implies $D_{\x}(P_{\x0})<1$. 

For any $\sigma_{\y\theta}^{*}$, $\mu_{\y}(\cdot)$, there exists $P_{\x1}$ large enough s.t.\ all consumers $v\geq c_{\bad}$ prefer firm $\y$ to $\x$ even at belief $\mu_{\x}=1$. Then firm $\x$, in particular type $\bad$, does not deviate from $P_{\x0}$ to $P_{\x}\geq P_{\x1}$, because $D_{\x}^{\mu}(P_{\x})\leq D_{\x}^{1}(P_{\x1})=0\;\forall \mu$. 

Due to $D_{\x}(P_{\x0})\in(0,1)$ and $\mu_{\x}(P_{\x0})<1$, we have $D_{\x}^{1}(P_{\x0})>D_{\x}(P_{\x0})$. Continuity of $D_{\x}^{1}(\cdot)$ is ensured, because the support of $\sigma_{j\good}^{*}$ does not include $P_{j}$ s.t.\ $\mu_{j}(P_{j})=1$. Thus for $\epsilon>0$ small enough, $D_{\x}^{1}(P_{\x0}+\epsilon)>D_{\x}(P_{\x0})$. 

Using the Mean Value Theorem and $(P_{\x 0}+\epsilon-c_{\bad})D^{1}_{i}(P_{\x 0}+\epsilon)> (P_{\x 0}-c_{\bad})D^{\mu_{\x}(P_{\x 0})}_{i}(P_{\x 0}) >(P_{\x 1}-c_{\bad})D^{1}_{i}(P_{\x 1})$, there exists $P_{*}\in(P_{\x0},P_{\x1})$ such that $(P_{*}-c_{\bad})D^{1}_{i}(P_{*}) =(P_{\x 0}-c_{\bad})D^{\mu_{\x}(P_{\x 0})}_{i}(P_{\x 0})$. Type $\bad$ strictly prefers the equilibrium price $P_{\x 0}$ to any $P>P_{*}$. For $\epsilon>0$ small enough, type $\good$ strictly prefers $P_{*}+\epsilon$ to $P_{\x 0}$, because the indifference of $\bad$ between $P_{\x 0}$ and $P_{*}$ implies $D^{1}_{i}(P_{*}) <D^{\mu_{\x}(P_{\x 0})}_{i}(P_{\x 0})$, thus $(P_{*}-c_{\bad})D^{1}_{i}(P_{*}) -(c_{\good}-c_{\bad})D^{1}_{i}(P_{*}) =(P_{*}-c_{\bad})D^{\mu_{\x}(P_{\x 0})}_{i}(P_{\x 0}) -(c_{\good}-c_{\bad})D^{1}_{i}(P_{*}) >(P_{\x 0}-c_{\bad})D^{\mu_{\x}(P_{\x 0})}_{i}(P_{\x 0})-(c_{\good}-c_{\bad})D^{\mu_{\x}(P_{\x 0})}_{i}(P_{\x 0})$. 
%
%
%
\end{proof}
An immediate corollary of Lemma~\ref{lem:elimpool3} is that equilibria where both firms pool are removed by the Intuitive Criterion. In such equilibria, at least one firm obtains positive demand without fully separating and the other firm's good type does not separate. 
Lemma~\ref{lem:elimpool3} stands in contrast to the constant quality premium case in Lemma~\ref{lem:asypool}, where all consumers switch firms at the same price, which makes profit discontinuous. Then there exist parameter values such that for any price and belief combination, either both types deviate to it or neither does. Simultaneous deviations mean that the Intuitive Criterion cannot eliminate the equilibrium.

\subsubsection{Constant quality premium}
\label{sec:const}

This section restricts attention to consumers who all have the same difference $h(v)-v=\nu$ in their valuations for good and bad quality. 

The following lemma shows that in any separating equilibrium, the bad type prices strictly above its marginal cost. The intuition is that otherwise the bad type would imitate $\good$ to make positive profit. 
\begin{lemma}
\label{lem:posprofitBertrand}
In any separating equilibrium, $\underline{P}_{i\bad}>c_{\bad}$, $\pi_{i\bad}^{*}>0$ and $\underline{P}_{i\good}>\overline{P}_{i\bad}+\nu$, and if $c_{\good}-c_{\bad}\leq \nu$, then $\underline{P}_{i\good}>c_{\good}$ and $\pi_{i\good}^{*}>0$ for $i\in\set{\x,\y}$. 
\end{lemma}
\begin{proof}
Assume $\underline{P}_{i\good}\leq \underline{P}_{j\good}$ w.l.o.g., so $D_{i}(\underline{P}_{i\good}) \geq \frac{\mu_{0}}{2}[1-F_{v}(\underline{P}_{i\good}-\nu)]>0$. Firm $i$'s type $\bad$ has the option to set $\underline{P}_{i\good}$ and make positive profit, so $\pi_{i\bad}^{*}>0$. Then $\pi_{j\bad}^{*}>0$, because $j\bad$ can set $P\in(c_{\bad},\underline{P}_{i\bad})$ and attract all customers from $i\bad$. Positive profit for $\bad$ implies $\underline{P}_{i\bad}>c_{\bad}$. 

Suppose $\underline{P}_{i\good}\leq \overline{P}_{i\bad}+\nu$, then any consumer who buys at $\overline{P}_{i\bad}$ also buys at $\underline{P}_{i\good}$. This implies $D_i^{1}(\underline{P}_{i\good})\geq D_i^{0}(\overline{P}_{i\bad})$, which motivates $\bad$ to imitate $\good$. 

If $c_{\good}-c_{\bad}\leq \nu$, then $\underline{P}_{i\good}> \overline{P}_{i\bad}+\nu>c_{\bad}+\nu\geq c_{\good}$, therefore $\pi_{i\good}^{*}>0$. 
\end{proof}

Lemma~\ref{lem:Bmonopoly} shows that the bad types price below their monopoly level. 
\begin{lemma}
\label{lem:Bmonopoly}
In any separating equilibrium, 
$\overline{P}_{i\bad}=\overline{P}_{j\bad}\leq P_{\bad}^{m}$ 
for $i\in\set{\x,\y}$. 
\end{lemma}
\begin{proof}
The result $\underline{P}_{i\good}>\overline{P}_{i\bad}+\nu$ in Lemma~\ref{lem:posprofitBertrand} implies that if $\underline{P}_{i\good}-\overline{P}_{j\bad}\geq \underline{P}_{j\good}-\overline{P}_{i\bad}$, then all consumers at $i\good$ who do not leave the market switch to $j\bad$ given the chance. 
If all customers at $i\good$ prefer to switch to $\overline{P}_{j\bad}$, then $\underline{P}_{j\bad}\geq \min\set{P_{\bad}^{m}, \underline{P}_{i\bad}}$, because at prices below $\underline{P}_{i\bad}$, firm $j$'s type $\bad$ is a monopolist and deviating only improves belief for $\bad$ in a separating equilibrium. Thus demand after raising price from $\underline{P}_{j\bad}$ is at least as great as for a monopolist known to be type $\bad$. 

By Lemma~\ref{lem:posprofitBertrand}, $\underline{P}_{j\good}>\overline{P}_{j\bad}+\nu$ and by the previous paragraph, $\underline{P}_{j\bad} \geq \underline{P}_{i\bad}$, so due to $\overline{P}_{j\bad}\geq \underline{P}_{j\bad}$, all consumers at $j\good$ strictly prefer $\underline{P}_{i\bad}$. 
Thus $\underline{P}_{i\bad}\geq \min\set{P_{\bad}^{m}, \underline{P}_{j\bad}}$. 

Suppose $\underline{P}_{i\bad}>P_{\bad}^{m}\leq \underline{P}_{j\bad}$. Deviating from $\underline{P}_{i\bad}$ to $P_{\bad}^{m}$ results in demand $D_{i}^{\mu_{i}(P_{\bad}^{m})}(P_{\bad}^{m})$ weakly above the monopoly level $D_{i}^{0}(P_{\bad}^{m})$. In addition, some consumers may switch to $j$ when facing $\underline{P}_{i\bad}$, but none switch at $P_{\bad}^{m}\leq \underline{P}_{j\bad}\leq \overline{P}_{j\bad}<\underline{P}_{j\good}-\nu$. Thus demand for $i\bad$ increases relatively more than in a complete information monopoly environment, where cutting price from $\underline{P}_{i\bad}$ to $P_{\bad}^{m}$ is profitable. The implication $\underline{P}_{i\bad} \leq P_{\bad}^{m} \leq \underline{P}_{j\bad}$ contradicts  $\underline{P}_{i\bad}>P_{\bad}^{m}\leq \underline{P}_{j\bad}$, therefore $\underline{P}_{i\bad}=\underline{P}_{j\bad}\leq P_{\bad}^{m}$. 
%

Separation implies positive profit by Lemma~\ref{lem:posprofitBertrand}, so $\underline{P}_{i\bad}>c_{\bad}$ and a small enough price cut does not lead to negative profit. 
Both firms' $\bad$ types choose an atomless price distribution in any separating equilibrium, because belief threats do not deter the bad type from undercutting an atom of the rival.  

At $\overline{P}_{i\bad}\geq \overline{P}_{j\bad}$, demand is only positive if the other firm is type $\good$. 
Deviating from $P_{i\bad} >P_{\bad}^{m}$ to $P_{\bad}^{m}$ increases demand from $\frac{1-\sigma_{j}^{\bad*}(P_{i\bad})}{2}D_{i}^{0}(P_{i\bad}) <\frac{1}{2}D_{i}^{0}(P_{i\bad})$ to $\frac{1}{2}D_{i}^{0}(P_{\bad}^{m}) \geq \frac{1}{2}D_{i}^{0}(P_{i\bad})$, which is a greater increase than under complete information monopoly. By assumption, monopoly profit under complete information is single peaked, so the demand increase from $P_{i\bad} >P_{\bad}^{m}$ to $P_{\bad}^{m}$ makes deviation profitable. Therefore $\overline{P}_{i\bad}\leq P_{\bad}^{m}$ for each firm. 
\end{proof}

%


Demand is positive on the support of prices, otherwise $\bad$ would deviate to $c_{\bad}+\epsilon$. Positive demand and $\pi_{\good}^{*}\geq 0$ imply that the support of the equilibrium strategies is weakly above $c_{\good}$. 
The assumption $\overline{v}>c_{\good}$ 
implies that the prices are strictly above $c_{\good}$, because otherwise $\good$ would deviate to $P=c_{\good}+\epsilon$ even at the worst belief $\mu_i(c_{\good}+\epsilon)=0$. 
Thus under incomplete information, price is always strictly above marginal cost for both types, which differs from a situation with public cost and quality, as shown in the next section.

\subsubsection{Comparison to public positively correlated quality and cost} 
\label{sec:completeinfoposcorrcompar}

Incomplete information may increase or decrease prices in Bertrand competition, as shown in this section for positive correlation of cost and quality (the case of negative correlation is in Section~\ref{sec:completeinfonegcorrcompar}). This indeterminate effect contrasts with costly search \citep{heinsalu2019a}, where asymmetric information greatly enhances competition (decreases price from monopoly to competitive) under negatively associated quality and cost, but reduces competition under positive correlation. 

Under complete information Bertrand competition with positively correlated cost and quality, if $c_{\good}-c_{\bad}>\nu$, then trade occurs at $P_{\good}=c_{\good}$ with probability $\mu_0^2$, at $P_{\bad}=c_{\good}-\nu$ with probability $2\mu_0(1-\mu_0)$ and at $P_{\bad}=c_{\bad}$ with probability $(1-\mu_0)^2$. If instead $c_{\good}-c_{\bad}<\nu$, then the trading price is $P_{\good}=c_{\good}$ with probability $\mu_0^2$, is $P_{\good}=c_{\bad}+\nu >c_{\good}$ with probability $2\mu_0(1-\mu_0)$ and $P_{\bad}=c_{\bad}$ with probability $(1-\mu_0)^2$. 

By contrast, incomplete information implies that trade occurs at a (semi)pooling price $P_0> c_{\good}$ (with probability at least $1-\mu_0^2$) or a semiseparating price $P_{i\good s}>P_0+(1-\mu_0)\nu$. 
The \emph{ex ante} expected price under incomplete information is higher than under complete iff either $c_{\good}-c_{\bad} \geq\nu$ or $\mu_0\notin(\underline{\mu}_0,\overline{\mu}_0)\subset(0,1)$. 
Prices are always strictly above marginal cost when the marginal cost and quality are private, but when these are public, then at least one firm charges a price equal to its marginal cost.

\subsubsection{Comparison to homogeneous consumers} 
\label{sec:homog}

If information is complete or cost and quality are negatively related (as in Section~\ref{sec:negcorr}), then the strategies of the firms against homogeneous and heterogeneous consumers are the same. 
However, under incomplete information and positive correlation of cost and quality,  \cite{janssen+roy2015} Proof of Proposition~2 initially claims the unique symmetric D1 equilibrium: 
\begin{enumerate}[(a)]
\item If 
$v_{\bad}>c_{\bad}+ v_{\good}-v_{\bad}$, 
then $P_{\good}=v_{\good} >v_{\bad}>c_{\bad}+v_{\good}-v_{\bad}$. 
\item If 
$v_{\bad}\leq c_{\bad}+ v_{\good}-v_{\bad}$, then 
$P_{\good}=\max\set{c_{\good},c_{\bad}+2(v_{\good}-v_{\bad})}$ and all consumers buy. 
\end{enumerate}
From the second paragraph on, \cite{janssen+roy2015} Proof of Proposition~2 says:  
\begin{enumerate}[(a)] 
\item If 
$v_{\bad}\geq c_{\bad}+ v_{\good}-v_{\bad}$, then $P_{\good}=c_{\bad}+2(v_{\good}-v_{\bad})$ and type $\bad$ mixes over $P_{\bad}\in[c_{\bad}+\mu_{0}(v_{\good}-v_{\bad}),\; c_{\bad}+v_{\good}-v_{\bad}]$. 
\item If 
$v_{\bad}< c_{\bad}+ v_{\good}-v_{\bad}$, then $P_{\good}=v_{\good}$ and type $\bad$ mixes over $[c_{\bad}+\mu_{0}(v_{\bad}-c_{\bad}),\; v_{\bad}]$. 
\end{enumerate}
Nonnegative profit for $\good$ requires $P_{\good}=c_{\bad}+2(v_{\good}-v_{\bad})\geq c_{\good}$, so incomplete information always increases $P_{\good}$, sometimes strictly. If $c_{\good}-c_{\bad}\leq v_{\good}-v_{\bad}$, then incomplete information strictly increases $P_{\bad}$, but if $c_{\good}-c_{\bad}> v_{\good}-v_{\bad}$ and 
$\mu_{0} <\frac{c_{\good}-c_{\bad}-v_{\good}+v_{\bad}}{v_{\bad}-c_{\bad}}$, 
then there is a positive probability that $P_{\bad}$ is lower under incomplete information. 
As $\mu_0\rightarrow 0$, the probability of trade at $P_{\bad}$ goes to $1$ in both cases. 


\subsection{Negatively correlated cost and quality}
\label{sec:negcorr}

In this section, the only differences from Section~\ref{sec:incompleteinfo} are that a firm with good quality has a lower cost, the gains from trade are positive for a bad quality firm, but not all consumers buy at the bad type's cost, and the complete information monopoly profit of the good type increases in price for prices below the bad type's cost. Formally, $c_{\good}<c_{\bad}<\overline{v}$, $c_{\bad}\geq h(0)$ and $\frac{d}{dP}P[1-F_{v}(h^{-1}(P))]>0$ for all $P\in[0,c_{\bad}]$. Normalise $c_{\good}=0$ w.l.o.g. 

Analogously to Lemma~\ref{lem:monot}, demand and price are monotone in type in any equilibrium, but due to $c_{\good}<c_{\bad}$, the direction of the monotonicity switches. 
Denote by $D_{i}(P)$ the equilibrium demand for firm $i$ at price $P$. 
\begin{lemma}
\label{lem:monotnegcorr}
In any equilibrium for any $P_{\theta}$ in the support of $\sigma_i^{\theta*}$, $D_i(P_{\good})\geq D_i(P_{\bad})$, and if $D_i(P_{\bad})>0$, then $P_{\good}\leq P_{\bad}$. 
\end{lemma}
\begin{proof}
If $P_{\bad}D_i(P_{\bad})-c_{\bad}D_i(P_{\bad}) \geq P_{\good}D_i(P_{\good})-c_{\bad}D_i(P_{\good})$ and $P_{\good}D_i(P_{\good})\geq P_{\bad}D_i(P_{\bad})$, then  ${\color{black}P_{\bad}D_i(P_{\bad})}-c_{\bad}D_i(P_{\bad}) \geq {\color{black}P_{\bad}D_i(P_{\bad})} -c_{\bad}D_i(P_{\good})$. 

In equilibrium, $D_i(P_{\bad})>0\Rightarrow P_{\bad}\geq c_{\bad}$, otherwise $\bad$ deviates to $P\geq c_{\bad}$. 

If $(P_{\bad}-c_{\bad})D_i(P_{\bad}) \geq (P_{\good}-c_{\bad})D_i(P_{\good})$ and $D_i(P_{\good})\geq D_i(P_{\bad})>0$, then $P_{\bad}-c_{\bad}\geq P_{\good}-c_{\bad}$. 
\end{proof}

The next lemma, similarly to Lemma~\ref{lem:withbeliefthreats}, rules out some equilibria even when belief threats are possible. 
\begin{lemma}
\label{lem:withbeliefthreatsnegcorr}
In any equilibrium, at least one firm's type $\good$ obtains positive profit. If type $\good$ of firm $i$ gets zero demand, then both types of $j$ set $P_{j}>c_{\bad}$ and receive positive demand and profit. 
If $\mu_0 h(0)<c_{\bad}$ or at least one firm partly separates or firms play symmetric strategies, then type $\good$ of each firm obtains positive profit. 
\end{lemma}
\begin{proof}
Suppose $\pi_{i\good}^*=0$ for both firms. Then for each firm and any $P_{i\good}$ in the support of $\sigma_{i}^{\good *}$, either $P_{i\good}=0$ or $D_i(P_{i\good})=0$. 
If $P_{i\good}=0\leq P_{j\good}$, then $D_i(0)>0$, so type $\bad$ of firm $i$ separates and sets $P_{i\bad}>0$. Then $j\good$ has probability $1-\mu_0$ of facing $i\bad$ with $P_{i\bad}>0$ and $\mu_{i}(P_{i\bad})=0$, so $j\good$ obtains positive demand and profit from $P\in(0,P_{i\bad})$ for any belief $\mu_{j}(P)$. Thus $P_{j\good}>0$ and $\pi_{j\good}^{*}>0$. 

Suppose $D_i(P_{i\good})=0$, then $D_i(P)=0$ for any $P>0$, otherwise $i\good$ would deviate to $P$ to get positive profit. Total demand $D_{\x}(P)+D_{\y}(P)$ is positive for any $P<\overline{v}$ for any beliefs $\mu_{\x}(P),\mu_{\y}(P)$, so if $D_i(P)=0$, then $D_{j}(P)>0$. Then both types of $j$ get positive profit from any $\hat{P}\in(c_{\bad},\overline{v})$, thus $j\bad$ sets $P_{j\bad}>c_{\bad}$. 

If $j\bad$ partly separates, i.e.\ sets $P_{j\bad}$ with $\mu_{j}(P_{j\bad})=0$, then both types of firm $i$ can get positive demand and profit by setting $P_{i}\in(c_{\bad},P_{j\bad})$, regardless of $\mu_{i}(P_{i})$. 

If $j\bad$ (semi)pools with $j\good$, then there exists $P_{j0}$ chosen by both $j\bad$ and $j\good$ with $\mu_{j}(P_{j0})\leq \mu_{0}$. 
If $\mu_0 h(0)< c_{\bad}$, then for $\epsilon>0$ small, consumers with valuations $v\in(0,\epsilon)$ strictly prefer to buy from $i$ at $P_{i}\in(0,v)$ and any belief, rather than from $j$ at $P_{j0}\geq c_{\bad}$ and $\mu_{j}(P_{j0})\leq \mu_{0}$. This makes $i\good$ deviate to $P_{i}\in(0,v)$ to get positive profit. 

If the firms play symmetric strategies, then consumers' beliefs are $\mu_{\x}(P)=\mu_{\y}(P)$ for any $P$ chosen in equilibrium, thus $D_{\x}(P)=D_{\y}(P)$. From $D_{\x}(P)+D_{\y}(P)>0$ for any $P<\overline{v}$, we get $P_{i\bad}\geq c_{\bad}$, otherwise $\bad$ would deviate to $P_{d}\geq c_{\bad}$ to obtain nonnegative profit. Type $\good$ can imitate $P_{i\bad}\geq c_{\bad}$ and receive $D_{i}(P_{i\bad})>0$, thus $\pi_{i\good}^*>0$. 
%
\end{proof}
Similarly to Lemma~\ref{lem:withbeliefthreats}, the results of Lemma~\ref{lem:withbeliefthreatsnegcorr} are tight. Asymmetric pooling (both types of firm $i$ set $P_{i0}>c_{\bad}$ and both types of $j$ set $P_{j0}\in(c_{\bad},P_{i0})$) ensures zero profit for both types of $j$. 
Symmetric pooling on $c_{\bad}$ guarantees zero profit for the bad types. 

Unlike under positive correlation of cost and quality, additional results do not require a constant quality premium $h(v)-v$. This is because all consumers prefer type $\good$ at $P_{\good}\leq P_{\bad}$. The Intuitive Criterion of \cite{cho+kreps1987} selects a unique equilibrium, as shown below. 
\begin{thm}
\label{thm:unique}
In the unique equilibrium passing the Intuitive Criterion, $P_{\bad}=c_{\bad}$ and type $\good$ mixes on $[\underline{P}_{\good},c_{\bad})$, where $\underline{P}_{\good}[1-F_{v}(h^{-1}(\underline{P}_{\good}))] =c_{\bad}(1-\mu_0)[1-F_{v}(h^{-1}(c_{\bad}))]$ and the cdf of price is 
\begin{align*}
\sigma_{i}^{\good *}(P) =\frac{1}{\mu_{0}}-\frac{c_{\bad}(1-\mu_0)[1-F_{v}(h^{-1}(c_{\bad}))]}{\mu_0P[1-F_{v}(h^{-1}(P))]}. 
\end{align*}
\end{thm}
\begin{proof}
First, rule out equilibria where some firm's type $\good$ gets zero profit. 
Lemma~\ref{lem:withbeliefthreatsnegcorr} proved that $\pi_{i\good}^{*}=P_{i\good}D_{i}(P_{i\good})>0$ for at least one firm. Due to $D_{i}(P_{i\good})\leq 1$, any $P_{i\good}$ in the support of $\sigma_{i}^{\good *}$ is bounded below by $\pi_{i\good}^{*}>0$. Denote by $\underline{P}_{i\theta}$ the lowest price in the support of $\sigma_{i}^{\theta *}$. 
To apply the Intuitive Criterion to show $\pi_{j\good}^{*}>0$ for $j\neq i$, set $\mu_{j}(P)=1$ for $P\in(0,\min\set{\underline{P}_{i\good},c_{\bad}})$. Then all consumers with $h(v)\geq P$ buy from firm $j$ at $P$. Suppose $\pi_{j\good}^{*}=0$, then $j\good$ strictly prefers $P$ to its equilibrium price, but $j\bad$ strictly prefers its equilibrium profit $\pi_{j\bad}^{*}\geq 0$ to the negative profit from $P<c_{\bad}$. This justifies $\mu_{j}(P)=1$ and removes any equilibrium with $\pi_{j\good}^{*}=0$. 

Second, rule out (partial) pooling. 
Positive profit for $\good$ implies positive demand, so pooling is only possible on $P_{i0}\geq c_{\bad}$, otherwise $\bad$ would deviate to nonnegative profit. 
Define $P_{*}:=\sup\set{P\leq P_{i0}:(P-c_{\bad})D_{i}^{1}(P)<\pi_{i\bad}^*}$. Due to $\pi_{i\good}^{*}>0$, we have $D_{i}(P_{0})>0$, so $P_{*}\geq c_{\bad}$. To apply the Intuitive Criterion to rule out (partial) pooling on $P_{i0}\geq c_{\bad}$, set $\mu_{i}(P)=1$ for all $P<P_{*}$. 
Because $i\bad$ puts positive probability on $P_{i0}$, we have $\mu_{i}(P_{i0})<1$. Then for $\epsilon>0$ small enough, $(P_{i0}-\epsilon)D^{1}(P_{i0}-\epsilon) >P_{i0}D_i(P_{i0})$, so $\good$ strictly prefers to deviate to $P_{i0}-\epsilon$, but $\bad$ strictly prefers the equilibrium. This justifies $\mu_{i}(P_{i0}-\epsilon)=1$ and removes (partial) pooling. 

Combining separation of types with Lemma~\ref{lem:monotnegcorr} shows $P_{i\good}<P_{i\bad}$. All consumers strictly prefer $P_{i\good}$ at $\mu_{i}(P_{i\good})=1$ to $P_{i\bad}$ at $\mu_{i}(P_{i\bad})=0$, so $i\bad$ only gets positive demand if firm $j$ has type $\bad$. 
The bad types Bertrand compete, so undercut each other to $P_{\bad}=c_{\bad}$. 

The good types mix atomlessly by the standard reasoning for Bertrand competition with captive customers. Atoms invite undercutting. The good types' price competition cannot reach $P=0$, because with positive probability, the other firm has type $\bad$ and sets $P_{\bad}=c_{\bad}$. This makes type $\good$ strictly prefer $P=c_{\bad}-\epsilon$ to $P\leq \epsilon$ for $\epsilon>0$ small enough. 

Denote by $\overline{P}_{i\theta}$ the supremum of the support of $\sigma_{i}^{\theta *}$. Combining separation, $P_{i\good}\leq P_{i\bad}$ and $P_{i\bad}=c_{\bad}$ yields $\overline{P}_{i\good}\leq c_{\bad}$ for both firms. 

Suppose $\overline{P}_{i\good}\leq \overline{P}_{j\good} <c_{\bad}$. If $\sigma_{i}^{\good*}$ has an atom at $\overline{P}_{i\good}$ and $\overline{P}_{i\good}=\overline{P}_{j\good}$, then set $\mu_{j}(\overline{P}_{j\good}-\epsilon)=1$ for $\epsilon>0$ small enough and $j\good$ will deviate from $\overline{P}_{j\good}$ to undercut $\overline{P}_{i\good}$. Type $\bad$ of firm $j$ strictly prefers not to set $\overline{P}_{j\good}-\epsilon$. Thus the Intuitive Criterion justifies $\mu_{j}(\overline{P}_{j\good}-\epsilon)=1$ and eliminates equilibria with an atom at $\overline{P}_{i\good}=\overline{P}_{j\good}<c_{\bad}$. 

If $\sigma_{i}^{\good*}$ has no atom at $\overline{P}_{i\good}$ or $\overline{P}_{i\good}<\overline{P}_{j\good}$, then demand at $\overline{P}_{j\good}$ is only positive when firm $i$ has type $\bad$. Set $\mu_{j}(P)=1$ for all $P\in[\overline{P}_{j\good},c_{\bad})$, then $D_{j}(P)>0$ iff firm $i$ has type $\bad$, in which case $j\good$ is a monopolist on $P\in[\overline{P}_{j\good},c_{\bad})$. By assumption, the complete information monopoly profit $PD^{1}(P)$ strictly increases on $[0,c_{\bad}]$, so $\good$ strictly prefers any $P\in(\overline{P}_{j\good},c_{\bad})$ to $\overline{P}_{j\good}$. Type $\bad$ strictly prefers not to set $P<c_{\bad}$. Therefore the Intuitive Criterion justifies $\mu_{j}(P)=1$ and eliminates equilibria with $\overline{P}_{i\good}\leq \overline{P}_{j\good} <c_{\bad}$. 

Suppose $\underline{P}_{i\good}< \underline{P}_{j\good}$, then set $\mu_{i}(P)=1$ for all $P\in[\underline{P}_{i\good}, \underline{P}_{j\good}]$. Then for any $P\in[\underline{P}_{i\good},\underline{P}_{j\good})$, all customers who end up buying buy from $i$. Because $PD^{1}(P)$ strictly increases on $[0,c_{\bad}]$, type $\good$ of firm $i$ will deviate from $\underline{P}_{i\good}$ to any $P\in[\underline{P}_{i\good}, \underline{P}_{j\good})$. This rules out $\underline{P}_{i\good}< \underline{P}_{j\good}$. 

Price $\underline{P}_{i\good}\leq \underline{P}_{j\good}$ attracts all customers with $h(v)\geq \underline{P}_{i\good}$. Profit from $\underline{P}_{i\good}$ is $\underline{P}_{i\good}[1-F_{v}(h^{-1}(\underline{P}_{i\good}))]$, which must equal the profit $c_{\bad}(1-\mu_0)[1-F_{v}(h^{-1}(c_{\bad}))]$ from $\overline{P}_{i\good} =c_{\bad}$. 
This determines $\underline{P}_{i\good}$. 
Type $\good$ must be indifferent between all $P\in[\underline{P}_{i\good},c_{\bad})$ in the support of $\sigma_{i}^{\good *}$, which determines firm $j$'s mixing cdf $\sigma_{j}^{\good *}$ by 
$P[1-\mu_0\sigma_{j}^{\good *}(P)][1-F_{v}(h^{-1}(P))]=c_{\bad}(1-\mu_0)[1-F_{v}(h^{-1}(c_{\bad}))]$. The same indifference condition for $j$ determines $\sigma_{i}^{\good *}$, so $\sigma_{\x}^{\good *}=\sigma_{\y}^{\good *}$. 
Solving for $\sigma_{i}^{\good *}$ yields 
$\sigma_{i}^{\good *}(P) =\frac{1}{\mu_{0}}-\frac{c_{\bad}(1-\mu_0)[1-F_{v}(h^{-1}(c_{\bad}))]}{\mu_0P[1-F_{v}(h^{-1}(P))]}$. 
\end{proof}

The equilibrium selected by the Intuitive Criterion in Theorem~\ref{thm:unique} is similar to the symmetric equilibrium in the homogeneous consumer case studied in \cite{janssen+roy2015} Lemma A.1 in which 
$P_{\bad}=c_{\bad}$ and type $\good$ mixes atomlessly over $[(1-\mu_0)c_{\bad}+\mu_0c_{\good},c_{\bad}]$. Thus trade occurs at $c_{\bad}$ with probability $(1-\mu_0)^2$, otherwise at some $P_{\good}<c_{\bad}$. The \emph{ex ante} expected price under incomplete information is lower than under complete iff $\mu_0\notin(\underline{\mu},\overline{\mu})\subset(0,1)$. 

The above symmetric equilibrium is not unique, contrary to the claim in \cite{janssen+roy2015} Lemma A.1. For example, both firms pooling on $P_{0}=c_{\bad}$ is an equilibrium when $c_{\bad}$ is below a cutoff. 
Similarly, with heterogeneous consumers, pooling on $c_{\bad}$ is an equilibrium for low $c_{\bad}$. However, it does not pass the Intuitive Criterion, regardless of whether valuations differ among customers or not. If belief is set to $1$ on prices below $c_{\bad}$, then the good types will deviate to $c_{\bad}-\epsilon$ and at least double their demand.

\subsubsection{Comparison to public negatively correlated quality and cost}
\label{sec:completeinfonegcorrcompar}

This section compares competition under incomplete and complete information when cost and quality are negatively associated. 
When both firms have the same type, both price at their marginal cost, as usual in symmetric Bertrand competition. Thus trade occurs at price $c_{\good}=0$ with probability $\mu_0^2$ (when both firms have type $\good$), but at $c_{\bad}$ with probability $(1-\mu_0)^2$. 
For asymmetric firms and a constant quality premium, trade occurs at $c_{\bad}+\nu$. The probability of unequal types is $2\mu_0(1-\mu_0)$.

By contrast, incomplete information results in $P_{\bad}=c_{\bad}$ and type $\good$ mixing atomlessly over $[\underline{P}_{\good},c_{\bad})$ for some $\underline{P}_{\good}>0$. Therefore trade occurs at $c_{\bad}$ with probability $(1-\mu_0)^2$, otherwise at some $P_{\good}<c_{\bad}$. The \emph{ex ante} expected price under incomplete information is lower than under complete iff $\mu_0\notin(\underline{\mu},\overline{\mu})\subset(0,1)$. 

Now assume that $h(v)-v$ increases in $v$, strictly for some $v$. Focus on asymmetric firms. 
Type $\bad$ does not set any $P<c_{\bad}$ attracting positive demand, thus type $\good$ is a monopolist on $P<h(c_{\bad})$. 
If the complete information monopoly price $P_{\good}^{m}$ 
of the good type is below $h(c_{\bad})$, then $\good$ sets $P_{\good}^{m}$, type $\bad$ gets zero demand and may set any $P\geq h^{-1}(P_{\good}^{m})$. This outcome is similar to the case of a constant quality premium. 

The more interesting case is $P_{\good}^{m}>h(c_{\bad})$,\footnote{Sufficient for $P_{\good}^{m}>h(c_{\bad})$ is $\frac{d}{dP}P[1-F_{v}(h^{-1}(P))]>0$ for all $P\in[0,h(c_{\bad})]$. 
} 
where type $\good$ raises price until the consumers with the lowest valuations above $c_{\bad}$ prefer to switch to $\bad$ charging $P_{\bad}=c_{\bad}$. These switchers are captive for type $\bad$, inducing it to raise price above $c_{\bad}$, which in turn loosens the incentive constraint of type $\good$, allowing it to raise price. The good type ends up pricing between $h(c_{\bad})$ and $P_{\good}^{m}$, and the bad type strictly above $c_{\bad}$. 
This result is similar to the above-competitive pricing found in Section~\ref{sec:completeinfo} for positively correlated public cost and quality, but differs from privately known cost and quality. 

Unlike with public types, private information about negatively correlated cost and quality leads to $P_{\good}<P_{\bad}=c_{\bad}$, with $P_{\good}>c_{\good}=0$, and if the firms have different types, then demand for the bad type is zero (Theorem~\ref{thm:unique}). The reason is that the good type signals its private quality by reducing price. 
Overall, price may end up higher or lower than under public cost and quality. 

\section{Conclusion}

Customer heterogeneity turns out to be important for optimal pricing decisions when better quality producers have higher costs, independently of private information about cost and quality. When higher quality costs less, firms price similarly when facing heterogeneous consumers as when homogeneous. Because information is valuable only if it changes decisions, a firm may estimate the value of gathering information on consumer preferences using its own and competitors' cost and quality data, if available. 

Asymmetric information about qualities and costs substantively affects pricing only when higher quality firms have a lower marginal cost. If smaller cost implies better quality, then a low price can credibly signal quality. This is because cutting price reduces the profit less for a firm with lower cost. In some situations, other costly signals of quality are feasible, for example advertising or warranties, but in other markets such as insurance, warranties are uncommon. Even if available, warranties or advertisements may not be the best way to signal---price cuts may be a cheaper or more accurate way for a high quality firm to distinguish itself. 

A low price used as a signal is similar to limit pricing, i.e.\ a monopolist keeping potential entrants out of the market by convincing them of its low cost by way of a low price. In limit pricing, charging a low price is anti-competitive, unlike in this paper. 

One policy implication of this paper is that a regulator maximising total or consumer surplus should discourage quality certification (and other methods to make quality public) when cost and quality are negatively correlated across firms and private. 
However, when better quality producers have higher costs, publicising quality may increase or decrease price and total and consumer surplus. Making information complete improves welfare \emph{ex ante} when firms' types are unknown if the cost difference between firms is expected to be large relative to the quality difference or consumers are nearly certain of the firms' types. Thus in mature industries, certification is likely to increase surplus. 

If the regulator knows the firms' costs and qualities, then welfare-maximisation suggests that with symmetric firms, this information should be revealed. If the firms differ and the better quality producer has a lower cost, then revelation reduces surplus. With asymmetric firms where the higher quality one has a greater cost, the price effect of making information public depends on the relative qualities and costs. 


The comparison between positively and negatively associated cost and quality suggests other policy implications. To maximise welfare under privately known cost and quality, the correlation between these should be made negative. Two ways to do this are to reward good quality (industry prizes for the best product) and punish for flawed products (fines, lawsuits). The quality of larger firms should be checked the most frequently to give them the greatest motive to improve it, because they are likely the low cost producers. 
Firms whose quality and cost are uncertain to consumers (e.g.\ start-ups) should receive targeted assistance with cost reduction if their quality is high, and with quality improvement  if their costs are low. 

If qualities and costs are negatively associated and private, then a merger to duopoly need not increase prices by much. Thus the optimal antitrust policy depends on the correlation of cost and quality in an industry.

\bibliographystyle{ecta}
\bibliography{teooriaPaberid} 
\end{document}